%% file: ne-npnp.tex
\documentclass{article}
\usepackage{graphicx,epsfig,color,latexsym}
\usepackage{wrapfig}
\usepackage{citesort}
\usepackage{enumerate}
\input{mac90}

\usepackage[top=1.in, bottom=1.in, left=1.0in, right=1.0in]{geometry}

\newcommand{\sparse}{{\rm SPARSE}}
\newcommand{\tally}{{\rm TALLY}}

\newcommand{\NT}{\NTIME}
\newcommand{\DT}{\DTIME}

\newcommand{\pad}{{\rm padding}}

\newcommand{\Pad}{{\rm Padding}}
\newcommand{\Ppoly}{{\rm P/Poly}}
\newcommand{\D}{{\rm Density}}

\newcommand{\nexpd}{\mbox{\rm {Nonexponentially-Dense-Class}}}

\begin{document}

\date{}

\title{NE is not NP Turing Reducible to Nonexpoentially Dense NP Sets
}

\author{
Bin Fu\thanks{This research is supported in part by
National Science Foundation Early Career Award 0845376.}
\\
\\
Department of Computer Science\\
 University of Texas--Pan American\\
Edinburgh, TX 78541, USA\\
Emails: binfu@cs.panam.edu
\\
 }
\maketitle

\begin{abstract}
A long standing open problem in the computational complexity theory
is to separate $\NE$ from $\BPP$, which is a subclass of $\NP_{\rm
T}({\NP})\cap\Ppoly$. In this paper, we show that $\NE\not\subseteq
\NP_{\rm T}(\NP\cap\nexpd)$, where $\nexpd$ is the class of
languages $A$ without exponential density (for each constant $c>0$,
$|A^{\le n}|\le 2^{n^c}$ for infinitely many integers $n$). Our
result implies $\NE\not\subseteq \NP_{\rm T}({\pad(\NP, g(n))})$ for
every time constructible super-polynomial function $g(n)$ such as
$g(n)=n^{\ceiling{\log\ceiling{\log n}}}$, where $\Pad(\NP, g(n))$
is class of all languages $L_B=\{s10^{g(|s|)-|s|-1}:s\in B\}$ for
$B\in \NP$. We also show $\NE\not\subseteq
\NP_{\T}(\P_{tt}(\NP)\cap\tally).$
\end{abstract}


\section{Introduction}
Separating the complexity classes has been one of the central
problems in complexity theory. Separating NEXP from P/Poly is a long
standing fundamental open problem in the computational complexity
theory. We do not even know how to separate NEXP from BPP, which is
a subclass of $\NP_{\rm T}({\NP})\cap\Ppoly$ proved by
Adleman~\cite{Adleman}.

Whether sparse sets are hard for complexity classes plays an
important role in the computational complexity theory (for examples,
\cite{BermanHartmanis,Mahaney,OgiwaraWatanabe01,KarpLipton}). It is
well known that $\Ppoly$ is the same as the class of languages that
are truth table reducible to tally sets ($\Ppoly=\P_{tt}(\tally)$).
The combination of bounded number of queries and density provides an
approach to characterize the complexity of the nonuniform
computation models. The partial progress for separating exponential
time classes from nonuniform polynomial time classes are shown
in~\cite{Watanabe87,Fu95,HarkinsHitchcock07,Hitchcock06,LutzMayordomo94}.
 Let $\nexpd$ be
the class of languages $A$ without exponential density (for each
constant $c>0$, $|A^{\le n}|\le 2^{n^c}$ for infinitely many
integers $n$). Improving Hartmanis and Berman's separation
$\E\not\subseteq \P_{m}(\nexpd)$~\cite{BermanHartmanis}, Watanabe
showed $\E\not\subseteq \P_{btt}(\nexpd)$. Watanabe's result was
improved by two research groups independently 
with incomparable results that $\E\not\subseteq
\P_{n^{1-\epsilon}-tt}(\nexpd)$ by Lutz and
Mayordomo~\cite{LutzMayordomo94}, and $\EXP\not\subseteq
\P_{n^{1-\epsilon}-\T}(\nexpd)$
 and $\E\not\subseteq
\P_{n^{{1\over 2}-\epsilon}-\T}(\nexpd)$ by Fu~\cite{Fu95}. Fu's
results were improved to $\E\not\subseteq
\P_{n^{1-\epsilon}-\T}(\nexpd)$ by Hitchcock~\cite{Hitchcock06}.  A
recent celebrated progress was made by Williams separating NEXP from
ACC~\cite{Williams10}. It is still an open problem to separate NEXP
from $\P_{O(n)-tt}(\tally)$.

The nondeterministic time hierarchy was separated in the early
research of complexity theory by Cook~\cite{Cook-nondet-hierarchy},
Serferas, Fischer, Meyer~\cite{SeiferasFischerMeyer78}, and
Zak~\cite{Zak83}.
 A separation with immunity among nondeterministic computational
complexity classes was derived by Allender, Beigel, Hertranpf and
Homer~\cite{AllenderBeigelHertrampfHomer}. The difference between NE
and NP has not been fully solved. One of the most interesting
problems between them is to separate NE from $\P_{\T}(\NP)$.
 Fu, Li and Zhong~\cite{FuLiZhong94} showed  $\NE\not\subseteq
\P_{n^{o(1)-\T}}(\NP)$. Their result was later improved by Mocas
\cite{Mocas96} to $\NEXP\not\subseteq \P_{n^c-\T}(\NP)$ for any
constant $c>0$. Mocas's result is optimal with respect to
relativizable proofs, as Buhrman and Torenvliet
\cite{BuhrmanTorenvliet94} showed an oracle relative to which
$\NEXP=\P_{\rm T}({\NP})$.  Buhrman, Fortnow and
Santhanam~\cite{BuhrmanFortnowSanthanam09} and Fu, Li and
Zhang~\cite{FuLiZhang09} showed $\NEXP=\P_{n^c-\T}({\NP})/n^c$ for
every constant $c>0$ (two papers appeared in two conferences with a
similar time). Fu, Li and Zhang showed that NEXP is not reducible to
tally sets by the polynomial time nondeterministic Turing reductions
with the number of queries bounded by a sub-polynomial function
$g(n)$ such as $g(n)=n^{1\over \log\log n}$ ($\NE\not\subseteq
\NP_{g(n)-\T}(\tally)$)\cite{FuLiZhang09}.

In this paper, we show that $\NE\not\subseteq \NP_{\rm
T}(\NP\cap\nexpd)$. Our result implies $\NE\not\subseteq \NP_{\rm
T}({\pad(\NP, g(n))})$ for every time constructible super-polynomial
function $g(n)$ such as $g(n)=n^{\ceiling{\log\ceiling{\log n}}}$,
where $\Pad(\NP, g(n))$ is the class of all languages
$L_B=\{s10^{g(|s|)-|s|-1}:s\in B\}$ for $B\in \NP$. We also show
$\NE\not\subseteq \NP_{\T}(\P_{tt}(\NP)\cap\tally).$

This paper is organized as follows. Some notations are given in
section~\ref{notations-sec}. In section~\ref{overview-sec}, we give
a  brief description of our method to prove the main result. In
section~\ref{main-theorem-sec}, we separate $\NE$ from $ \NP_{\rm
T}(\NP\cap\nexpd)$. In section~\ref{NP-low-sec}, we show how to use
the padding method to derive sub-exponential density problems in the
class $\NP$. In section~\ref{NE-tally-sec}, we separate $\NE$ from $
\NP_{\rm T}(\P_{tt}(\NP)\cap\tally)$. The conclusions are given in
section~\ref{conclusion-sec}.

\section{Notations}\label{notations-sec}
Let $N=\{0,1,2,\cdots\}$ be the set of all natural numbers. Let
$\Sigma=\{0,1\}$ be the alphabet for all the languages in this
paper. The length of a string $s$ is denoted by $|s|$. Let $A$ be a
language. $A^{\le n}$ is the subset of strings of length at most $n$
in $A$. $A^{=n}$ is the subset of strings of length $n$ in $A$.
 For a finite set $X$, let $|X|$ be the number of elements in $X$.
 For a Turing machine $M(.)$, let $L(M)$ be the language accepted by
 $M$. We use a pairing function $(.,.)$ with $|(x,y)|=O(|x|+|y|)$.

For a function $t(n):N\rightarrow N$,  let $\DT(t(n))$ be the class
of languages accepted by deterministic Turing machines in $O(t(n))$
time, and $\NT(t(n))$ be the class of languages accepted by
nondeterministic Turing machines in $O(t(n))$ time. Define the
exponential time complexity classes: $\E=\cup_{c=1}^{\infty}
\DT(2^{cn})$, $\EXP=\cup_{c=1}^{\infty} \DT(2^{n^c})$,
$\NE=\cup_{c=1}^{\infty} \NT(2^{cn})$ and $\NEXP=\cup_{c=1}^{\infty}
\NT(2^{n^c})$.

A language $L$ is {\it sparse} if for some constant $c>0$, $|L^{\le
n}|\le n^c$ for all large $n$.  Let $\sparse$ represent all sparse
languages. Let $\tally$ be the class of languages with alphabet
$\{1\}$.


Assume that $M(.)$ is an oracle Turing machine.  A decision
computation
 $M^A(x)$ returns either $0$ or $1$ when the input is $x$ and
oracle is $A$.

Let $\le^{\P}_r$ be a type of polynomial time reductions, and
$\textsl{S}$ be a class of languages. $\P_r(\textsl{S})$ is the
class of languages $A$ that are reducible to some languages to
$\textsl{S}$ via $\le^P_r$ reductions. In particular, $\le_m^{\P}$
is the polynomial time {\it many-one reduction}, and $\le_{\T}^{\P}$
is the polynomial time {\it Turing reduction}.

For a class $C$ of languages, we use $\NP_{\rm T}(C)$ to represent
the class of languages that can be reducible to the languages in $C$
via polynomial time nondeterministic Turing reductions.

For a nondecreasing function $d(n):N\rightarrow N$, define
$\D(d(n))$ to be the class of languages $A$ with $|A^{\le n}|\le
d(n)$ for all sufficiently large $n$.

For a function $f(n): N\rightarrow N$, it is {\it time
constructible} if given $n$, $f(n)$ can be computed in $O(f(n))$
steps by a deterministic Turing machine.

A function $d(n):N\rightarrow N$ is {\it nonexponential} if for
every constant $c>0$, $d(n)< 2^{n^c}$ for infinitely many integers
$n$. $\nexpd$ is the class of languages $A$ whose density function
$d_A(n)=|A^{\le n}|$ is nonexponential.

\section{Overview of Our Method}\label{overview-sec}

We give a brief description about our method in this section. Our
main theorem is proved by contradiction. Assume that $\NEXP\subseteq
\NP_{\T}(S)$, where $S$ is a language in both $\NP$ and $\nexpd$.
Since $S$ is not of exponential density, we can find a function
nondecreasing unbounded function $e(1^n)$ that is computable in
$2^{n^{O(1)}}$ time and satisfies $|S^{\le n}|\le 2^{n^{1\over
e(1^n)^2}}$ for infinitely many integers $n$. Let $h(n)=n^{e(1^n)}$.
Thus, $h(n)$ is super-polynomial function.

Our main technical contribution is a counting method to be combined
with the classical translational method in deriving the separation.
Select an arbitrary language $L_0$ in $\DT(2^{h(n)})$. We define the
language $L_1=\{x10^{h(|x|)-|x|-1}: x\in L_0\}$. This converts $L_0$
into a language in $\NEXP$. Using the assumption $\NEXP\subseteq
\NP_{\T}(S)$, we have a polynomial time oracle Turing machine $M_1$
to accept $L_1$ with oracle $S$.

Define another language $L_2=\{1^n0m: m\le 2^n$ and there are at
least $m$ different strings $z_1,\cdots, z_m$ that are queried by
$M_1$ with some input of length $h(n)\}$. We can also show that
$L_2$ is also in $\NEXP$.  When $S$ has a subexponential number of
elements with length at most $h(n)^{O(1)}$, we show that the largest
$m$ with $1^n0m\in L_2$ has  $m< 2^n$.

In the next, we spend $2^{n^{O(1)}}$ time to find the largest $m$,
which will be denoted by $m_n$. This can be easily done since $L_2$
is in $\NP_{\T}(\NP)$.

For $m_n$ with $m_n< 2^n$, consider a nondeterministic computation
that given an input $(x,m_n)$ with $n=|x|$, it guesses all the
strings $z_1,\cdots, z_{m_n}$, which are queried by $M_1$ by inputs
of length $h(n)$, of $S$ in a path. Thus, any query like $y\in S?$
is identical to check if $y$ is equal to one of elements in
$z_1,\cdots, z_{m_n}$. This is an nondeterministic computation of
exponential time. It can be converted into a problem in
$\NP_{\T}(\NP)$. It can be simulated in a deterministic
$2^{n^{O(1)}}$ time. Since there are infinitely many integers $n$
with
 $|S^{\le n}|\le
2^{n^{1\over e(1^n)^2}}$, we have infinitely many integers
$n_1,n_2,\cdots$ to meet this case with $m_{n_i}< 2^{n_i}$.
This brings a $2^{n^{O(1)}}$ time deterministic Turing machine $M_*$
that $L_0^{=n_i}=L(M_*)^{=n_i}$ for some
 for
infinitely many integers $n_i$. We can construct $L_0$ in
$\DT(2^{h(n)})$ to make it impossible using the standard diagonal
method. This brings a contradiction.

\section{Main Separation Theorem}\label{main-theorem-sec}

In this section, we present our main separation theorem. The theorem
is achieved by the translational method, which is combined with a
counting method to count the number of all possible strings queried
by nondeterministic polynomial time oracle Turing machine.

\begin{definition}\label{query-def}\scrod
\begin{itemize}
\item
Let $M$ be an oracle nondeterministic Turing machine. Let $a_1\cdots
a_{i-1}$ be a $0,1$-sequence,  and $y$ be an input for $M$. Define
$H(M(y), a_1\cdots a_{i-1})$ to be the set of all  strings $z$ that
are queried by $M(y)$ at the $i$-th time at some path assuming $M$
receives answers $a_1,\cdots, a_{i-1}$ for its first $i-1$ queries
from the oracle (the answer for each query is either `0' or `1' from
the oracle).
\item
For a nondeterministic oracle Turing machine $M(.)$ and oracle $A$,
and an integer $k$, define $Q(M, A, k)$ to be the set all strings
$z$ in $A$ such that $z\in H(M(y), a_1\cdots a_{i-1})$  for some
string $y$ of length $k$ and some $a_1\cdots a_{i-1}\in \{0,1\}^*$.
\end{itemize}
\end{definition}

\begin{lemma}\label{ne-nexp-lemma}
Let $\Gamma$ be a class of languages and be closed under
$\le_m^{\P}$-reductions. Then $\NE\subseteq \Gamma$ if and only if
$\NEXP\subseteq \Gamma$.
\end{lemma}

\begin{proof}
Since $\NE\subseteq \NEXP$, it is trivial that $\NEXP\subseteq
\Gamma$ implies $\NE\subseteq \Gamma$. We only prove that
$\NE\subseteq \Gamma$ implies $\NEXP\subseteq \Gamma$. Assume
$\NE\subseteq \Gamma$. Let $L$ be an arbitrary language in $\NEXP$.
Assume that $L\in \NT(2^{n^c})$ for some integer constant $c>1$. Let
$L'=\{x10^{|x|^c-|x|-1}:x\in L\}$. Since $L\in \NT(2^{n^c})$ with
the constant $c$, we have $L'\in\NE$. We have a
$\le^{\P}_m$-reduction $f(.)$ from $L$ to $L'$ with
$f(x)=x10^{|x|^c-|x|-1}$ ($L\le_m^{\P} L'$). Since $L'\in
\NE\subseteq \Gamma$ and $\Gamma$ is closed under
$\le_m^{\P}$-reductions, we have $L\in \Gamma$. Since $L$ is an
arbitrary language in $\NEXP$, we have $\NEXP\subseteq \Gamma$.
\end{proof}

\begin{lemma}\label{nondet-lemma}
Let $M_*(.)$ be a nondeterministic polynomial time oracle Turing
machine. Let $A$ be a language in $\NP$ and accepted by a polynomial
time Turing machine $M_A(.)$.  Then there is a nondeterministic
$mn^{O(1)}$ time Turing machine $N(.)$ such that given the input
$(m, M_*, M_A, 1^n)$,
\begin{itemize}
\item
 if $m\le |Q(M_*,A,n)|$, it outputs a subset of $m$ different
elements of $Q(M_*,A,n)$ in at least one path, and every path with
nonempty output gives a subset of $m$ different elements of
$Q(M_*,A,n)$; and
\item
 if $m> |Q(M_*,A,n)|$, it outputs empty set in every path.
\end{itemize}
\end{lemma}

\begin{proof}
Let $M_A(.)$ be a polynomial time nondeterministic Turing machine
that accepts $A$, and run in time $n^{c_A}$ for a constant $c_A>0$.
 Let $M_*(.)$
have time bound $n^{c_*}$. We design a nondeterministic Turing
machine $N(.)$.

Let $N(.)$ do the following with input $(m, M_*, M_A, 1^n)$:

\begin{enumerate}[1.]
\item
\qquad guess strings $x_1,\cdots, x_m$ of length $n$;
\item
\qquad guess a path $p_i$ and a series of oracle answers
$a_{i,1}\cdots a_{i,j_i-1}$ for $M_*(x_i)$ for $i=1,\cdots, m$;
\item\label{H-query-line}
\qquad if $M_*(x_i)$ makes the $j_i$-th query $z_i$ on path $p_i$
assuming the first the $j_i-1$ oracle answers are $a_{i,1}\cdots
a_{i,j_i-1}$;
\item
\qquad then guess a path $q_i$ for $M_A(z_i)$
\item
\qquad if $z_1,\cdots, z_m$ are all different, and each $z_i$ is
accepted by $M_A(z_i)$ on path $q_i$
\item
\qquad then output $z_1,\cdots, z_m$
\item
\qquad else output the empty set $\emptyset$.
\end{enumerate}
We note that line~\ref{H-query-line} is to check if $z_i$ is in
$H(M_*(x_i), a_{i,1}\cdots a_{i,j_{i-1}})$. Since $M_*(.)$ runs in
time $n^{c_*}$, each $z_i$ is of length at most $n^{c_*}$. The
Turing machines $M_A(z_i)$ takes $|z_i|^{c_A}\le n^{c_*c_A}$ time to
accept $z_i$ for $i=1,\cdots,m$. Therefore, the total time of $N(.)$
with input $(m, M_*, M_A, 1^n)$ is $mn^{O(1)}$.

\end{proof}



\begin{lemma}\label{e(n)-lemma}
Assume that $S$ is in  $\NP$ and $S$ is nonexponentially dense. Then
there is a $2^{n^{O(1)}}$ time computable nondecreasing function
$e(1^n):N\rightarrow N$ such that
\begin{enumerate}[1.]
\item\label{square-root-condition-in-e(n)-lemma}
 $|S^{\le n}|\le 2^{n^{1\over e(1^n)^2}}$ for infinitely many integers $n$;
\item\label{square-condition-in-e(n)-lemma}
$e(1^{n^2})\le 2e(1^n)$ for all $n$; and
\item
$\lim_{n\rightarrow\infty}e(1^n)=\infty$.
\end{enumerate}
\end{lemma}

\begin{proof}
Let $e(1^0)=1$. We construct $e(1^n)$ at phase $n$. Assume that we
have constructed $e(1^1),\cdots, e(1^{t-1})$. Phase $t$ below is for
computing $e(1^t)$.

\vskip 10pt

 Phase $t$

\begin{enumerate}[1).]
\item
\qquad Let $k$ be the largest number less than $t$ with
$e(1^{k-1})<e(1^k)$.

\item\label{square-slow}
\qquad If $t\le k^2$, then let $e(1^t)=e(1^k)$, and enter Phase
$t+1$.

\item\label{t-exp-line}
\qquad If $t\not=j^{(e(1^k)+1)^2}$ for any integer $j$, then let
$e(1^t)=e(1^k)$, and enter Phase $t+1$.

\item\label{compute-s}
\qquad Compute $s=|S^{\le t}|$.

\item\label{square-condition}
\qquad If $s\le 2^{t^{1\over (e(1^k)+1)^2}}$, then let
$e(1^t)=e(1^k)+1$.

\end{enumerate}

End of Phase $t$.

\vskip 10pt

The purpose of line~\ref{t-exp-line} is to let $t=j^{(e(1^k)+1)^2}$
for some integer $j$ after this line. This makes $t^{1\over
(e(1^k)+1)^2}$ be an integer and makes the computation easy at
line~\ref{square-condition}. Checking the condition of the if
statement at line~\ref{t-exp-line} takes $t^{O(1)}$ time via a
binary search.
 Computing $s$ at step~\ref{compute-s} in Phase $t$ takes
$2^{t^{O(1)}}$ steps since $S\in \NP$. Thus,
 function $e(1^n)$ is computable in
$2^{n^{O(1)}}$ time. Since $S$ is nonexponentially dense, the if
condition in step~\ref{square-condition} can be eventually satisfied
and we have that $e(1^n)$ is unbounded.

 Step~\ref{square-condition} in Phase $t$ makes function $e(.)$
satisfy condition~\ref{square-root-condition-in-e(n)-lemma} in the
lemma.  Step~\ref{square-slow} and  Step~\ref{square-condition} in
Phase $t$ makes function $e(.)$ satisfy
condition~\ref{square-condition-in-e(n)-lemma} in the lemma. The
construction shows that $e(1^n)$ is nondecreasing since $e(1^t)\le
e(1^{t+1})$ for all integers $t$.

\end{proof}

\begin{lemma}\label{infinitely2-cases-lemma}
Assume that $t(1^n)$ is nondecreasing unbounded function and
$t(1^n)$ is computable in $2^{n^{O(1)}}$ time. Then there is a
language $L_0\in \DT(2^{n^{t(n)}})$ such that for every
deterministic Turing machine $M(.)$ in time $2^{n^{O(1)}}$,
$L(M)^{=n}\not= L_0^{=n}$ for all sufficiently large $n$.
\end{lemma}

\begin{proof}
Let $M_1,\cdots, M_k,\cdots$ be the list of all deterministic Turing
machines  that  each $M_k$ runs in at most $2^{n^{t(1^n)/3}}$ time
for all large $n$. The construction has infinitely phases for
$n=1,2,\cdots$. It is easy to see that for each $2^{n^{O(1)}}$ time
Turing machine $N(.)$, there is a $2^{n^{t(1^n)/3}}$ time Turing
machine $M_i(.)$ with $L(M_i)^{=n}=L(N)^{=n}$ for all large $n$.

Phase $n$:


\qquad Let $x_1,\cdots, x_n$ be the first $n$ $0,1$-strings of
length $n$ by the lexicographic order

\qquad  For $i=1,\cdots, n$, put $x_i$ into $L_0^{=n}$ if and only
if $L(M_i)(x_i)$ rejects.

End of Phase $n$.

According to the construction of phase $n$. The language $L_0$ can
be computed in deterministic time $n\cdot 2^n\cdot
2^{n^{t(n)/3}}<2^{n^{t(n)/2}}$ for all large $n$. By the
construction of $L_0$,  for each Turing machine $M_i$ that runs in
time $2^{n^{t(1^n)/3}}$, $L(M_i)^{=n}\not=L_0^{=n}$ for all large
$n$.

\end{proof}

Theorem~\ref{main-theorem} and Theorem~\ref{ne-main-theorem} are
basically equivalent. They are the main separation results achieved
in this paper. We will find more concrete complexity classes inside
$\NP\cap \nexpd$ in section~\ref{NP-low-sec}.

\begin{theorem}\label{main-theorem}
$\NEXP\not\subseteq \NP_{\rm T}(\NP\cap \nexpd)$.
\end{theorem}

\begin{proof}
Assume $\NEXP\subseteq \NP_{\rm T}(\NP\cap \nexpd)$. We will bring a
contradiction from this assumption. Since $\NEXP$ has a complete
language $K$ under $\le_m^{\P}$ reductions, if $K\in \NP_{\rm
T}(S)$, then $\NEXP\subseteq \NP_{\rm T}(S)$. Let $S$ be a language
in $\NP\cap \nexpd$ such that
\begin{eqnarray}
\NEXP\subseteq \NP_{\rm T}(S). \label{assumption-inclusion}
\end{eqnarray}

 By Lemma~\ref{e(n)-lemma}, we have a
nondecreasing unbounded function $e(1^n)$ that satisfies
\begin{eqnarray}
e(1^{n^2})\le 2e(1^n) \label{e(n)-e(n^2)-ineqn}
\end{eqnarray}
 and $(|S^{\le n}|)\le 2^{n^{1/e(1^n)^2}}$ for
infinitely many integers $n$. Furthermore, function $e(1^n)$ is
computable in $2^{n^{O(1)}}$ time. Let
\begin{eqnarray}
h(n)=n^{e(1^n)}.\label{h-def0-eqn}
\end{eqnarray}

 We apply the translational method to it. Let $L_0$ be an
arbitrary language in $\DT(2^{h(n)})$, and accepted by a
deterministic Turing machine $N(.)$ in $\DT(2^{h(n)})$ time. Define
$L_1=\{x10^{h(|x|)-|x|-1)}: x\in L_0\}$.


Since function $e(1^n)$ is computable in $2^{n^{O(1)}}$ time, it is
easy to see that $L_1$ is in $\EXP\subseteq \NEXP$. By our
assumption (\ref{assumption-inclusion}), there is a nondeterministic
polynomial time oracle Turing machine $M_1(.)$ for $L_1\in \NP_{\rm
T}(S)$ (In other words, $M_1^S(.)$ accepts $L$). Assume that
$M_1(.)$ runs in time $n^{c_1}$ for all $n\ge 2$. Let $2\le
u_1<u_2<\cdots <u_k<\cdots$ be the infinite list of integers such
that
\begin{eqnarray}
d_S(u_i)=|S|^{\le u_i}\le 2^{u_i^{1/e(1^{u_i})^2}}.
\label{S-ui0-ineqn}
\end{eqnarray}

Define the language $L_2=\{1^n0m: m\le 2^n$ and there are at least
$m$ different strings $z_1,\cdots, z_m$ in $Q(M_1,S, h(n))\}$. Let
$n_i$ be the largest integers at least $2$ such that
\begin{eqnarray}
h(n_i)^{c_1}\le u_i \label{h(ni)-ui0-ineqn}
\end{eqnarray}
 for all large integers $i\ge i_0$ (it is easy to see the existence of such an integer $i_0$). Thus, we have
\begin{eqnarray}
h(n_i+1)^{c_1}> u_i. \label{h(ni+1)-ineqn}
\end{eqnarray}

For all large  integers $i$, we have
\begin{eqnarray}
e(1^{n_i})\ge 8c_1 \label{e(ni)-8c1-ineqn}
\end{eqnarray}
 since $e(1^n)$ is nondecreasing and unbounded.
 Since $S$ is of density bounded by $d_S(n)$, the number of strings in $S$ queried by $M_1(.)^S$ with inputs of
length $h(n)$ is at most $d_S(h(n)^{c_1})$. In other words, we have
\begin{eqnarray}
|Q(M_1,S, h(n))|\le d_S(h(n)^{c_1}).\label{Q-M1-ineqn}
\end{eqnarray}
 For the case $n=n_i$, we have
the inequalities:
\begin{eqnarray}
d_S(h(n_i)^{c_1})&\le& d_S(u_i)\ \ \ \ \ \ \ \mbox{(by\ inequality\ (\ref{h(ni)-ui0-ineqn}))}\label{ds-1}\\
&\le& 2^{u_i^{1/e(1^{u_i})^2}}\ \ \ \ \ \ \ \mbox{(by\ inequality\ (\ref{S-ui0-ineqn}))}\\
&\le& 2^{(h(n_i+1)^{c_1})^{1/e(1^{u_i})^2}}\ \ \ \ \ \ \ \mbox{(by\ inequality\ (\ref{h(ni+1)-ineqn}))}\\
&\le&2^{h(n_i^2)^{c_1/e(1^{u_i})^2}}\ \ \ \ \ \ \ \mbox{(by\ the\ condition\ $n_i\ge 2$)}\\
&\le&2^{(n_i^{2e(1^{n_i^2})})^{c_1/e(1^{n_i})^2}}\ \ \ \ \ \ \ \mbox{(by\ equation\ (\ref{h-def0-eqn}))}\\
&\le&2^{(n_i^{4e(1^{n_i})})^{c_1/e(1^{n_i})^2}}\ \ \ \ \ \ \ \mbox{(by\ inequality\ (\ref{e(n)-e(n^2)-ineqn}))}\\
&<& 2^{n_i} \ \ \ \ \ \ \ \mbox{(by\ inequality\
(\ref{e(ni)-8c1-ineqn}))}\label{ds-1z}
\end{eqnarray}

By inequalities  (\ref{ds-1}) to (\ref{ds-1z}), and
(\ref{Q-M1-ineqn}), we have the inequality
\begin{eqnarray}
|Q(M_1,S, h(n_i))|<2^{n_i}\ \ \ \ \ \ \mbox{for\ all\ large\
$i$}.\label{Q-M1-2-ni-ineqn}
\end{eqnarray}

 By Lemma~\ref{nondet-lemma}, $L_2$ is in $\NEXP$. By our
assumption (\ref{assumption-inclusion}), $L_2\in \NP_{\rm T}({S})$
via some nondeterministic polynomial time oracle Turing machine
$M_2(.)$. Assume that $M_2(.)$ runs in time $n^{c_2}$ for all $n\ge
2$, where $c_2$ is a positive constant.


Define the language $L_3=\{(x,m):$ $m\le  2^{|x|}$ and there are at
least $m$ different strings $z_1,\cdots, z_m$ in $Q(M_1,S, h(n))$,
and $M_1(x10^{h(|x|)-|x|-1)})$ has an accept path that receives
answer $1$ for each query (to oracle $S$) in $\{z_1,\cdots, z_m\}$,
and answer $0$ for each query (to oracle $S$) not in $\{z_1,\cdots,
z_m\}$ $\}$.

By Lemma~\ref{nondet-lemma}, we have $L_3\in \NE$.  Thus, $L_3\in
\NP_{\rm T}({S})$ via another  nondeterministic  polynomial time
oracle Turing machine $M_3(.)$. Assume that $M_3(.)$ runs in time
$n^{c_3}$ for all $n\ge 2$.

In order to find the largest number $m$ such that $1^n0m\in L_2$,
$m$ is always at most $2^{n}$. Thus, the length of $m$ is at most
$n+1$. Using the binary search, we can find the largest $m_{n_i}$
with $1^{n_i}0m_{n_i}\in L_2$ for $i=1,2,\cdots$. Let $m_{n_i}$ be
the largest $m$ with $1^{n_i}0m\in L_2$ for $i=1,2,\cdots$. Since
$S\in \NP$,  $m_{n_i}$ can be computed in $2^{{n_i}^{c_4}}$ time for
some positive constant $c_4$ for all $i=1,2,\cdots$. By
inequalityies~(\ref{Q-M1-2-ni-ineqn}), we have $m_{n_i}<2^{n_i}$.

\vskip 10pt
 {\bf Claim 1.} For $|x|=n_i$, we have
$x10^{h(n_i)-n_i-1}\in L_1$ if and only if $(x,m_{n_i})\in L_3$.
\vskip 10pt

\begin{proof} Assume that $z_1,\cdots,
z_{m_{n_i}}$ are different elements in $Q(M_1,S, h(n))$. By the
definition of $m_{n_i}$,
 a query if $y\in S$ made by $M_1^S(x10^{h(n_i)-n_i-1})$
to the oracle $S$ is identical to checking if $y\in \{z_1,\cdots,
z_{m_{n_i}}\}$. This is because all the strings in $S$ that are
queried are in the list $z_1,\cdots, z_{m_{n_i}}$. Thus,
$x10^{h(n_i)-n_i-1}\in L_1$ if and only if $(x,m_{n_i})\in L_3$.
\end{proof}

Assume that $m_{n_i}$ is known. We just check if $(x,m_{n_i})\in
L_3$ with $|x|=n_i$. For $|x|=n_i$, we have  $x\in L_0$ if and only
if $x10^{h(n_i)-n_i-1}\in L_1$ if and only if $(x,m_{n_i})\in L_3$
by Claim 1. Since $L_3\in \NP_{\rm T}({S})$ and $S\in\NP$, we only
need $2^{n^{c_5}}$ time to decide if $(x,m_n)\in L_3$ for
$n=n_1,n_2,\cdots$, where $c_5$ is a positive constant. Therefore,
we can decide if $x\in L_0$ in $2^{{n_i}^{c_5}}$ time for $|x|=n_i$.
Therefore, there is a deterministic Turing machine $M_*$ that runs
in $2^{n_i^{c_5}}$ time and has $L(M_*)^{=n_i}=L_0^{=n_i}$ for all
$i$ sufficiently large. Since $L_0$ is an arbitrary language in
$\DT(2^{{h(n)}})$. Function $h(n)$ is a super-polynomial function.
This brings there is a deterministic Turing machine $M_*$ that runs
in $2^{n^{c_5}}$ time and has $L(M_*)^{=n_i}=L_0^{=n_i}$ for all
sufficiently large $i$, which contradicts
Lemma~\ref{infinitely2-cases-lemma}.

\end{proof}

\begin{theorem}\label{ne-main-theorem}
$\NE\not\subseteq \NP_{\rm T}(\NP\cap \nexpd)$.
\end{theorem}
\begin{proof}
It follows from  Lemma~\ref{ne-nexp-lemma} and
Theorem~\ref{main-theorem}.
\end{proof}

\begin{corollary}
$\NEXP\not\subseteq \NP_{\rm T}({\NP\cap \sparse})$.
\end{corollary}

Although it is hard to achieve $\NEXP\not=\P_{\T}(\NP)$ or
$\NEXP\not\subseteq \P_{\T}(\sparse)$, we still have the following
separation.

\begin{corollary}
$\NEXP\not\subseteq \P_{\rm T}({\NP\cap \sparse})$.
\end{corollary}

\section{Hard Low Density Problems in NP}\label{NP-low-sec}
It is natural to ask if there exists any hard low density problem in
the class \NP. In this section, we show the existence of low density
sets in class NP. They are constructed from all natural NP-hard
problems under the well known exponential time hypothesis that
$\NP\not\subseteq \DT(2^{n^{o(1)}})$~\cite{ImpagliazzoPaturi99}.

\begin{definition}\label{well-super-sub-def}\scrod
\begin{itemize}
\item
A function $g(n):N\rightarrow N$ is {\it super-polynomial} if for
every constant $c>0$, $g(n)\ge n^c$ for all large $n$.
\item
A function $f(n):N\rightarrow N$ is {\it sub-polynomial} if for
every constant $c>0$, $f(n)\le n^c$ for all large $n$.
\item
A function $g(n):N\rightarrow N$ is called well-super-polynomial if
$g(n)$ is super-polynomial,  $g(n)$ is time constructible,  and
there is a time constructible sub-polynomial function $f(n)$
 such that  $f(g(n))\ge n$ for all sufficiently  large  $n$.
\item
A function $f(n):N\rightarrow N$ is called well sub-polynomial if
$f(n)$ is sup-polynomial,  $f(n)$ is time constructible, and there
is another time constructible super-polynomial function $h(n)$ such
that for each positive constant $c$,  $f(h(n)^c)\le n$ for all
sufficient large $n$.
\end{itemize}
\end{definition}

Define $\log^{(1)} n=\log n=\ceiling{\log_2 n}$. For integer $k\ge
1$, define $\log ^{(k+1)} n=\log(\log^{(k)} n)$.

We provide the following lemma to give some concrete slowly growing
well-sub-polynomial and well-super-polynomial functions.

\begin{lemma}\label{concrete-well-functions-lemma}\scrod
\begin{enumerate}[1.]
\item\label{state1}
For each constant integer $k>1$ and constant integer $a\ge 1$, the
function $\ceiling{n^{1/(\log^{(k)}n})^a}$ is time constructible
function from $N\rightarrow N$.
\item\label{state2}
For each constant integer $k>1$ and constant integer $a\ge 1$, the
function $n^{(\log^{(k)}n)^a}$ is time constructible function from
$N\rightarrow N$.
\item\label{state3}
Assume $k$ and $a$ are fixed integers with $k>1$ and $a> 1$.  Let
$f(n)=\ceiling{n^{1/(\log^{(k)}n)^a}}$ and
$h(n)=n^{(\log^{(k)}n)^{a-1}}$, then $f(h(n))<n^{o(1)}$ for all
large $n$.
\item\label{state4}
Assume $k$ and $a$ are fixed integers with $k\ge 1$ and $a\ge 1$.
Let $f(n)=\ceiling{n^{1/(\log^{(k)}n)^a}}$ and
$g(n)=n^{(\log^{(k)}n)^{a+1}}$, then $f(g(n))>n$ for all large $n$.
\end{enumerate}
\end{lemma}

\begin{proof}
Statement~\ref{state1}: It takes $O(\log n)$ time to compute
$\log^{(k)}n$. It takes another $O(\log n)$ time to compute
$(\log^{(k)}n)^a$ since $a$ is a constant. It takes another $O(\log
n)$ time to compute $\ceiling{n^{1/(\log^{(k)}n)^a}}$ via binary
search. Since $\log n=o(\ceiling{n^{1/(\log^{(k)}n)^a}})$, we have
that the function $\ceiling{n^{1/(\log^{(k)}n)^a}}$ is time
constructible.

Statement~\ref{state2}: It takes $O(\log n)$ time to compute
$\log^{(k)}n$. It takes another $O(\log n)$ time to compute
$m=(\log^{(k)}n)^a$ since $a$ is a constant. Using the elementary
method for multiplication,  we can compute $n^m$ with $O(m(\log
n^m)^2)=O(m^2\log n)=o(n^{(\log^{(k)}n)^a})$ time. Therefore,
$n^{(\log^{(k)}n)^a}$ is time constructible.

Statement~\ref{state3}: We have
\begin{eqnarray*}
f(h(n))&=& f(n^{(\log^{(k)}n)^{a-1}})\\
&=&n^{O({1\over \log^{(k)}n})}\\
 &<&n \ \ \ \mbox{for\ all\ large\ }n.
\end{eqnarray*}

Statement~\ref{state4}: We have
\begin{eqnarray*}
f(g(n))&=& f(n^{(\log^{(k)}n)^{a+1}})\\
&=&n^{\Omega(\log^{(k)}n)}\\
 &>&n \ \ \ \mbox{for\ all\ large\ }n.
\end{eqnarray*}

\end{proof}

\begin{definition}\scrod
\begin{itemize}
\item
For a language $A$, let $\pad(A, g(n))$ is the languages
$L=\{x10^{g(|x|)-|x|-1}:x\in A\}$.
\item
For a class $\Lambda$ of languages, define $\Pad(\Lambda, g(n))$ to
be the class of languages $\pad(A, g(n))$ for all $A\in \Lambda$.
\end{itemize}
\end{definition}

For example, let $g(n)=n^{(\log\log n)^k}$ for a fixed integer $k>1$
and let $f(n)=n^{1\over (\log\log n)^{k-1}}$. We have
$2^{f(g(n))})\ge 2^n$ for all sufficient large $n$.

\begin{definition}\label{subexponential-density-def}
A language $A$ is of subexponential density if for each constant
$c>0$, $|A^{\le n}|\le 2^{n^c}$ for all large $n$.
\end{definition}

\begin{lemma}\label{density0-lemma}
Assume that $A$ is a language and $g(n)$ is a super-polynomial
function. Then $\pad(A,g(n))$ is language of subexponential density.
\end{lemma}

\begin{proof}
For each language $A$, there are at most $2^n$ strings of length $n$
in $A$. When $s$ is mapped into $s10^{g(|s|)-|s|-1)}$, its length
becomes $g(|s|)$. Let $c$ be an arbitrary positive constant. As
$g(n)$ is a super-polynomial function, there is a constant integer
$n_c\ge 2$ such that for every $n>n_c$,
\begin{eqnarray}
g(n)>n^{10/c}.\label{g-n-ineqn}
\end{eqnarray}
Let $m_c=n_c^{10\over c}$. We have $n_c=m_c^{c\over 10}$. Let $m$ be
an arbitrary number greater than $m_c$. Let $k$ be the largest
integer with $g(k)\le m$.

Case 1: $k\le n_c$. The number of strings of length at most $k$ is
at most $2\cdot 2^k\le 2^{2k}\le 2^{2n_c}<2^{n_c^2}\le
2^{m_c^{c\over 5}}<2^{m^{c\over 5}}$. Therefore, the number of
strings of length at most $m$ in $\pad(A,g(n))$ is at most
$2^{m^{c\over 5}}$.

Case 2: $k>n_c$. We have $m\ge g(k)>k^{10\over c}$ by
inequality~(\ref{g-n-ineqn}). Thus, $k<m^{c\over 10}$. The number of
strings of length at most $k$ at is no more than $2\cdot
2^k<2^{2k}<2^{k^2}<2^{m^{c\over 5}}$. Therefore, the number of
strings of length at most $m$ in $\pad(A,g(n))$ is at most
$2^{m^{c\over 5}}$.

In every case, we have $|\pad(A,g(n))^{\le m}|\le 2^{m^{c\over 5}}$.
Since $c$ is an arbitrary positive constant,  $\pad(A,g(n))$ is a
language of subexponenital density by
Definition~\ref{subexponential-density-def}.
\end{proof}

\begin{lemma}\label{density-lemma}
Assume that $A$ is a language and $g(n)$ is a strictly increasing
super-polynomial function and $f(n)$ is a sub-polynomial function
with $f(g(n))\ge n$, then $\pad(A,g(n))$ is language of density
$O(2^{f(n)})$.
\end{lemma}

\begin{proof}
For each language $A$, there are at most $2^n$ strings of length $n$
in $A$. When $s$ is mapped into $s10^{g(|s|)-|s|-1)}$, its length
becomes $g(|s|)$. Since $g(n)$ is increasing super-polynomial
function, $g(n)<g(n+1)$ for all large $n$. We have $2^n\le
2^{f(g(n))}$. Thus, $\pad(A,g(n))$ is language of density
$O(2^{f(n)})$.
\end{proof}

We have Theorem~\ref{padding-theorem} that shows the existence of
subexponential density sets that are still far from polynomial time
computable under the reasonable assumption that $\NP\not\subseteq
\DT(2^{n^{o(1)}})$.

\begin{theorem}\label{padding-theorem}
Assume that $g(n)$ is a strictly increasing well-super-polynomial
function, and $f(n)$ is a sub-polynomial function with $f(g(n))\ge
n$. If $\NP\not\subseteq \DT(2^{n^{o(1)}})$, then for every
$\NP$-complete language $A$, $\pad(A, g(n))$ is a language of
density of $\D(2^{f(n)})$, and not in  $\DT(T(n))$, where $T(n)$ is
an arbitrary function with $T(g(n))=2^{n^{o(1)}}$.
\end{theorem}

\begin{proof} Let $A$ be a NP-complete language.
The density of $\pad(C, g(n))$ follows from
Lemma~\ref{density-lemma}. If $\pad(C, g(n))$ is computable in time
$T(n)$, we have that $A$ is computable in time
$T(g(n))=2^{n^{o(1)}}$. Thus, $\NP\subseteq \DT(2^{n^{o(1)}})$. This
contradicts the condition $\NP\not\subseteq \DT(2^{n^{o(1)}})$.
\end{proof}

The following corollary gives concrete result by assigning concrete
functions for $f(n), g(n)$ and $T(n)$.

\begin{corollary}
Let $g(n)=n^{(\log^{(k)})^a}$, $f(n)=\ceiling{n^{1\over (\log^{(k)}
n)^{a-1}}}$, and $T(n)=2^{\ceiling{n^{1/(\log^{(k)}n)^{a+1}}}}$ with
fixed integers $a>1$ and $k>1$. If $\NP\not\subseteq
\DT(2^{n^{o(1)}})$, then for every $\NP$-complete language $A$,
$\pad(A, g(n))$ is a language of density of $\D(2^{f(n)})$, and not
in  $\DT(T(n))$.
\end{corollary}

\begin{proof}
For $g(n)=n^{(\log^{(k)}n)^a}$  and  $f(n)=n^{1\over (\log^{(k)}
n)^{a-1}}$. By statement~\ref{state4} of
Lemma~\ref{concrete-well-functions-lemma}, we have $f(g(n))\ge n$.

For $T(n)=2^{\ceiling{n^{1/(\log^{(k)}n)^{a+1}}}}$ for an arbitrary
constant $c>0$. By statement~\ref{state3} of
Lemma~\ref{concrete-well-functions-lemma}, we have
$T(g(n))=2^{n^{o(1)}}$. The three functions satisfy the conditions
in Theorem~\ref{padding-theorem}. The corollary follows from
Theorem~\ref{padding-theorem}.
\end{proof}


We separate both $\NEXP$ and $\NE$ from $\NP_{\rm T}({\pad(\NP,
g(n))})$ for any    super-polynomial time constructible function
$g(n)$ from $N$ to $N$ in Theorems~\ref{next-np-padding-theorem} and
\ref{ne-np-padding-theorem}. For a given $g(n):N\rightarrow N$,
$\NP_{\rm T}({\pad(\NP, g(n))})$ is a concrete computational
complexity class.

\begin{theorem}\label{next-np-padding-theorem}
 Assume that $g(n)$ is a super-polynomial time constructible
function from $N$ to $N$. Then $\NEXP\not\subseteq \NP_{\rm
T}({\pad(\NP, g(n))})$.
\end{theorem}

\begin{proof}
It follows from Lemma~\ref{density0-lemma} and
Theorem~\ref{main-theorem}.
\end{proof}

\begin{theorem}\label{ne-np-padding-theorem}
 Assume that $g(n)$ is a super-polynomial time constructible
function from $N$ to $N$. Then $\NE\not\subseteq \NP_{\rm
T}({\pad(\NP, g(n))})$.
\end{theorem}

\begin{proof}
It follows from Lemma~\ref{ne-nexp-lemma} and
Theorem~\ref{next-np-padding-theorem}.
\end{proof}



\section{Separating $\NEXP$ from $\NP_{\rm T}({\P_{tt}(\NP)\cap
\tally})$}\label{NE-tally-sec}

In this section, we separate $\NEXP$ from $\NP_{\rm
T}({\P_{tt}(\NP)\cap \tally})$. A more generalized theorem is given
by Theorem~\ref{tally-theorem}. We are more carefully to combine the
counting method with the translational method to prove it.

\begin{definition}\scrod
\begin{itemize}
\item
 Let $M_1$ be a nondeterministic oracle Turing machine and $M_2$
be a deterministic oracle Turing machine.
 Define $M_1^{M_2}$ be a nondeterministic Turing machine such
that $M_1(x)$ takes an input $x$, each query $y$ produced by $M_1$
is answered by $M_2(y)$, which will access an oracle during the
computation.
\item
 Let $M_1$ be a nondeterministic oracle Turing machine and $M_2$
be a deterministic oracle Turing machine. Let $A$ be an oracle for
$M_2$. Define $(M_1^{M_2})^A$ be a nondeterministic Turing machine
$M_1^{M_2}$ with oracle $A$ such that $M_1(x)$ takes an input $x$,
each query $y$ produced by $M_1$ is answered by $M_2^A(y)$.
\end{itemize}
\end{definition}

\begin{definition}\scrod
\begin{itemize}
\item
For an  oracle Turing machine $M$ and an integer $k$, define $PQ(M,
y, k)$ to be the union of all $H(M(y), a_1\cdots a_{i-1})$ (see
Definition~\ref{query-def}) with $i\le k$ and $a_1\cdots a_{i-1}\in
\{0,1\}^{\le k}$.
\item
Assume that $M_1$ is a nondeterministic Turing machine and $M_2$ is
a deterministic adaptive oracle Turing machine $M(.)$. Let  $A$ be
an oracle set, and $k$ is an integer. Define
\begin{eqnarray}
Q_1(M_1^{M_2}, A, B, k_1,k_2, m)=\bigcup_{z\in \left(\bigcup_{y\in
B^{=m}}PQ(M_1,y,k_1)\right)} (A\cap PQ(M_2, z,
k_2)).\label{Q1-def-eqn}
\end{eqnarray}
\end{itemize}
\end{definition}

The purpose of Lemma~\ref{det-lemma} for the proof of
Theorem~\ref{tally-theorem} is similar to Lemma~\ref{nondet-lemma}
for Theorem~\ref{main-theorem}.

\begin{lemma}\label{det-lemma}
Assume that $A\in\NP$, $B\in \NP$
and $M_1(.)$ and $M_2()$ are polynomial time nondeterministic Turing
machines. Then there is a nondeterministic machine $N(.)$ such that
given the input $(M_1^{M_2}, M_A, M_B, k_1, k_2, 1^n)$, if $m\le
|Q_1(M_1^{M_2},A,B, k_1,k_2,n)|$, it outputs a subset of $m$
different elements of $Q_1(M_1^{M_2},A,B,k_1, k_2, n)$ in time
$mn^{O(1)}$ in at least one path; and otherwise, it outputs empty
set in every path, where $M_A$ is an polynomial time
nondeterministic Turing machine to accept $A$, and $M_B$ is a
polynomial time nondeterministic Turing machine to accept $B$.
\end{lemma}

\begin{proof}
 We design a nondeterministic
Turing machine $N(.)$.
Let $N(.)$ do the following with input $(M_*, M_A, M_B, k_1, k_2,
1^n)$:

\begin{enumerate}[1.]
\item
guess strings $x_1,\cdots, x_m$ of length $n$,
\item
guess a path $h_i$ of $M_B(x_i)$ for each $x_i$,
\item
guess a path $p_i$ and a query $y_i$ for each $M_1(x_i)$,
\item
guess a path $w_i$ and a query $z_i$ for each $M_2(y_i)$, and
\item
 guess a path $q_i$ for $M_A(z_i)$ for $i=1,\cdots, m$.
\item
 If $M_B(x_i)$ accepts in path $h_i$, $M_1(x_i)$
queries $y_i$ in path $p_i$ for $i=1,\cdots, m$, $M_2(y_i)$ queries
$z_i$ in path $w_i$ for $i=1,\cdots, m$, and $M_A(z_i)$ accepts in
path $q_i$ for $i=1,\cdots, m$, then $N$ outputs all $z_1,\cdots,
z_m$. Otherwise, $N$ outputs $\emptyset$.
\end{enumerate}

Note that for a path $p_i$ and a query $y_i$ for $M_1(x_i)$, a part
of path $p_i$ is  $a_1\cdots a_{j-1}, j$ with $j\le k_1$ such that
$M_1(x_i)$ follows
 path $p_i$ and its $j$-th query is $y_i$ assuming it
receives the $j-1$ answers are $a_1\cdots a_{j-1}$.

Note that for a path $w_i$ and a query $z_i$ for $M_2(y_i)$, a part
of path $w_i$ is $b_1\cdots b_{j-1}, j$ with $j\le k_2$ such that
$M_2(y_i)$ follows a path $w_i$ and its $j$-th query is $z_i$
assuming it receives the $j-1$ answers are $b_1\cdots b_{j-1}$.

Since $M_1(.), M_2(.), M_A(.)$ and $M_B(.)$ all run in polynomial
time, we have that the time for $N(.)$ is bounded by $mn^{O(1)}$.
\end{proof}

\begin{definition}
For a set $B$, define $\wp(B)$ to be the {\it power set} of $B$ (the
class of all subsets of $B$).
\end{definition}

Theorem~\ref{tally-theorem} gives another separation for $\NEXP$
from the polynomial time hierarchy. It is incomparable with
Theorem~\ref{main-theorem}.

\begin{theorem}\label{tally-theorem}
Assume that $B$ is an language in $(\NP\cap \coNP)\cap \nexpd$. Then
for any well sub-polynomial function $g(n)$,  $\NEXP\not\subseteq
\NP_{\rm T}(\P_{g(n)-\T}(\NP)\cap \wp(B))$.
\end{theorem}

\begin{proof}
We use a combination of counting method and translational method to
prove this theorem. Let $M_B$ be a polynomial time nondeterministic
Turing machine to accept $B$, and $M_{\overline{B}}$ be a polynomial
time nondeterministic Turing machine to accept $\overline{B}$. Let
$\SAT$ be the well known $\NP$-complete problem.

Assume $\NEXP\subseteq \NP_{\rm T}(\P_{g(n)-\T}(\NP)\cap \wp(B))$.
Since $\NEXP$ has a complete language under $\le^{\P}_m$-reductions,
we assume $\NEXP\subseteq \NP_{\rm T}({K})$ for some $K\subseteq B$
and also $K\in \P_{g(n)-\T}(\NP)$.
Let $K\in \P_{g(n)-\T}(\SAT)$ via oracle Turing machine $M_q(.)$.
Let $n^{c_q}$ be the running time of $M_q$.  Since $B\in \NP\cap
\coNP$ and is of nonexponential density, by Lemma~\ref{e(n)-lemma},
we have $e(1^n)$ to be a nondecreasing function with
$\lim_{n\rightarrow \infty} e(1^n)=\infty$,
\begin{eqnarray}
e(1^{n^2})\le 2e(1^n),\label{double-ineqn}
\end{eqnarray}
 and $(|B^{\le n}|)\le 2^{n^{1/e(1^n)^2}}$ for
infinitely many integers $n$. Furthermore, function $e(1^n)$ is
computable in $2^{n^{O(1)}}$ time.


Since $g(n)$ is a well sub-polynomial function, let $h_g(n)$ be a
well super-polynomial function (see
Definition~\ref{well-super-sub-def}) such that for each positive
constant $c$,
\begin{eqnarray}
g(h_g(n)^c)\le n\ \ \ \  \ \mbox{for\ all\ large\ integers\
$n$.}\label{hg0-def-eqn}
\end{eqnarray}
Let
\begin{eqnarray}
h(n)=\min(n^{e(1^n)}, h_g(n), 2^n).\label{h-def-eqn}
\end{eqnarray}

 We apply the translational method to it. Let $L$ be an
arbitrary language in $\DT(2^{{h(n)}})$. Define $L_1=\{x10^{h(n)}:
x\in L\}$.

Since $e(1^n)$ is computable in $2^{n^{O(1)}}$ time and $h_g(n)$ is
time constructible, we have that $L_1$ is in NEXP. Let
$L_1\in\NP_{\rm T}({K})$ via an nondeterministic oracle Turing
machine $M_1(.)$ with oracle $K$. Assume that $M_1(.)$ runs in time
$n^{c_1}$ for all $n\ge 2$. Let $u_1<u_2<\cdots <u_k<\cdots$ be the
infinite list of integers at least $2$ such that
\begin{eqnarray}
d_B(u_i)=(|B|^{\le u_i})\le
2^{u_i^{1/e(1^{u_i})^2}}.\label{ui-d-ineqn}
\end{eqnarray}
 Let $n_i$ be the
largest integers at least $2$ such that
\begin{eqnarray}
h(n_i)^{c_1}\le u_i \label{h(ni)-ui-ineqn}
\end{eqnarray}
 for all sufficiently large integers $i$. Thus, we have
\begin{eqnarray}
h(n_i+1)^{c_1}> u_i. \label{h(ni+1)-ui-ineqn}
\end{eqnarray}

Define $L_2=\{1^n0m: m\le 2^{2n}$ there are $m$ different elements
in $Q_1(M_1^{M_q},\SAT,B,h(n)^{c_1},g (h(n)^{c_1}), h(n))$ $\}$.

The number of strings $z\in B$ queried by $M_1(.)$ with inputs of
length $h(n)$ is at most $d_B(h(n)^{c_1})$ since the length of $z$
is at most $h(n)^{c_1}$. In other words,
\begin{eqnarray}
\left|\bigcup_{y\in \Sigma^{=h(n)}}PQ(M_1,y,h(n)^{c_1})\right|\le
d_B(h(n)^{c_1}).\label{PQ-M1-ineqn}
\end{eqnarray}

As $M_q$ is a deterministic oracle Turing machine with the number of
queries bounded by function $g(.)$, we have
\begin{eqnarray}
|PQ(M_q,z,g(h(n)^{c_1})|\le 2^{g(h(n)^{c_1})}\ \ \ \ \mbox{ for\
each\ $z$\ of\ length\ at\ most\ $h(n)^{c_1}$.}\label{PQ-Mq-ineqn}
\end{eqnarray}

By equations~(\ref{Q1-def-eqn}), (\ref{PQ-M1-ineqn}), and
(\ref{PQ-Mq-ineqn}), we have the inequality
\begin{eqnarray}
|Q_1(M_1^{M_q},\SAT,B,h(n)^{c_1},g (h(n)^{c_1}), h(n))|\le
d_B(h(n)^{c_1})2^{g(h(n)^{c_1})}\ \ \ \ \ \ \mbox{
 for\  all\ large\ $n$.} \label{Q1-almost-ready-ineqn}
\end{eqnarray}

 We will show
this number is less than $2^{2n_i}$ if $n=n_i$ for all large $i$.
For all large  $i$, we have
\begin{eqnarray}
e(1^{n_i})\ge 8c_1 \label{e(ni)-8c1-ineqn2}
\end{eqnarray}
 since $e(1^n)$ is nondecreasing and unbounded.
 Since $B$ is of density bounded by $d_B(n)$, we have the
 inequalities
\begin{eqnarray}
d_B(h(n_i)^{c_1})&\le& d_B(u_i)\ \ \ \ \ \mbox{(by\ inequality\ (\ref{h(ni)-ui-ineqn}))}\\
&\le& 2^{u_i^{1/e(1^{u_i})^2}}\ \ \ \ \ \mbox{(by\ inequality\
(\ref{ui-d-ineqn}))}\\
&\le& 2^{(h(n_i+1)^{c_1})^{1/e(1^{u_i})^2}}\ \ \ \ \ \mbox{(by\ inequality\ (\ref{h(ni+1)-ui-ineqn}))}\\
&\le&2^{h(n_i^2)^{c_1/e(1^{u_i})^2}}\ \ \ \ \ \mbox{(by\ the\ condition $n_i\ge 2$)}\\
&\le&2^{(n_i^{2e(1^{n_i^2})})^{c_1/e(1^{n_i})^2}}\ \ \ \ \ \mbox{(by\ equation\ (\ref{h-def-eqn}))}\\
&\le&2^{(n_i^{4e(1^{n_i})})^{c_1/e(1^{n_i})^2}}\ \ \ \ \ \mbox{(by\
equation\ (\ref{double-ineqn}))}\\
 &<& 2^{n_i}. \ \ \ \ \ \mbox{(by\
inequality\ (\ref{e(ni)-8c1-ineqn2}))}\label{db-query-ineqn}
\end{eqnarray}

Therefore,
\begin{eqnarray}
d_B(h(n_i)^{c_1})2^{g(h(n_i)^{c_1})}&\le& d_B(h(n_i)^{c_1})\cdot
2^{n_i}\ \ \ \  \ \mbox{(by\ inequality (\ref{hg0-def-eqn}) and equation (\ref{h-def-eqn}))}\\
&<&2^{2n_i}. \ \ \ \  \ \mbox{(by\ inequality (\ref{db-query-ineqn})
)}\label{total-queries-ineqn}
\end{eqnarray}

By inequalities~(\ref{Q1-almost-ready-ineqn}), and
(\ref{total-queries-ineqn})
\begin{eqnarray}
|Q_1(M_1^{M_q},\SAT,B,h(n_i)^{c_1},g (h(n_i)^{c_1}), h(n_i))|<
2^{2n_i}\ \ \ \ \ \ \mbox{
 for\  all\ large\ $i$.}
\end{eqnarray}

We can assume that $m\le 2^{2n}$ (otherwise, $1^n0m\not\in L_2$).
Since $M_1(.)$ and $M_q(.)$ run in $n^{c_1}$ and $n^{c_q}$ time,
respectively, we have that $M_1^{M_q}(.)$ runs in $n^{c_1c_q}$ time.
 By Lemma~\ref{det-lemma}, the decision if $1^n0m\in L_2$ can be
 made by a nondeterministic Turing machine in $mh(n)^{O(1)}= 2^{O(n)}$ time for all large $n$.
 We have $L_2\in \NEXP$. Thus, $L_2\in \NP_{\rm T}(K)$
via a nondeterministic Turing machine $M_2(.)$. Since $K\in
\P_T^{\NP}$, there is a constant $c_2$ such that we can find the
largest $m_n$ in time $2^{n^{c_2}}$ such that $1^n0m\in L_2$.


Define the language $L_3=\{(x,m):$ there are at least $m$ different
strings $z_1,\cdots, z_m$ in\\
 $Q_1(M_1^{M_q}, \SAT, B,h(n)^{c_1},
g(h(n)^{c_1}), h(n))$, and $(M_1^{M_q})^{\SAT}(x10^{h(|x|)-|x|-1)})$
has an accept path that receives answer $1$ for each query (to
\SAT), which is generated by some $y\in B$, in $\{z_1,\cdots,
z_m\}$, and answer $0$ for each query (to \SAT), which is generated
by some $y\in B$, not in $\{z_1,\cdots, z_m\}$ $\}$.

By
Lemma~\ref{nondet-lemma}, we have $L_3\in \NE$. Thus, $L_3\in
\NP_{\rm T}({K})$ via another polynomial time nondeterministic
Turing machine $M_3(.)$. Assume that $M_3(.)$ runs in time $n^{c_3}$
for all $n\ge 2$.

Assume $n_i=|x|$. In order to find the largest number $m$ such that
$1^{n_i}0m\in L_2$, $m$ is always at most $2^{2n_i}$. Thus, the
length of $m$ is at most $2n$. Using the binary search, we can find
the largest $m$ with $1^{n_i}0m\in L_2$. Let $m_{n_i}$ be the
largest $m$ with $1^{n_i}0m\in L_2$. Since $\SAT\in \NP$, $m_{n_i}$
can be computed in $2^{n_i^{c_4}}$ deterministic time for some
positive constant $c_4$.

Assume that $m_{n_i}$ is known. We just check if $(x,m_{n_i})\in
L_3$, where $n_i=|x|$. It is easy to see $x\in L$ if and only if
$x10^{h(n_i)-n_i-1}\in L_1$ if and only if $(x,m_{n_i})\in L_3$.
Since $L_3\in \NP_{\rm T}({K})$ and $K\in \P_T^{\NP}$, we only need
$2^{n_i^{c_5}}$ time to decide if $(x,m_{n_i})\in L_3$, where $c_5$
is a positive constant. Therefore, we can decide if $x\in L$ in
$2^{n_i^{c_5}}$ time.

Therefore, there is a deterministic Turing machine $M_*(.)$ that
runs in $2^{n_i^{c_5}}$ time and accepts $L^{n_i}$ for all $i$
sufficiently large. Note that $L$ is an arbitrary language in
$\DT(2^{{h(n)}})$, and function $h(n)$ is a super-polynomial
function. This brings there is a deterministic Turing machine
$M_*(.)$ that runs in $2^{n^{c_5}}$ time and accepts $L^{n_i}$ for
all sufficiently large integers $i$. This contradicts
Lemma~\ref{infinitely2-cases-lemma}.
\end{proof}

It is easy to see that $\{1\}^*$ is a sparse language in
$\P\subseteq \NP\cap \coNP$, and $\tally$ is the power set of
$\{1\}^*$. We have the following corollary.

\begin{corollary}
$\NEXP\not\subseteq \NP_{\rm T}({\P_{g(n)-\T}(\NP)\cap \tally})$ for
any well sub-polynomial function $g(n)$.
\end{corollary}

It is well known that $\P_{tt}(\NP)=\P_{O(\log n)-\T}(\NP)$
(see~\cite{BussHay
}), we have
corollary~\ref{tt-tall-corollary}.

\begin{corollary}\label{tt-tall-corollary}
$\NEXP\not\subseteq \NP_{\rm T}({\P_{tt}(\NP)\cap \tally})$.
\end{corollary}

\section{Conclusions}\label{conclusion-sec}
 We show that $\NEXP\not\subseteq \NP_{\rm T}({\NP\cap \nexpd})$. This result has almost reached the limit of relativizable technology.
   A fundamental open problem is to separate NEXP from BPP. We
would like to see the further step toward this target. Our method is
a relativizable. Since there exists an oracle to collapse NEXP to
BPP by Heller~\cite{Heller}, separating $\NEXP$ from $\BPP$ requires
a new way to go through the barrier of relativization. We feel that
it is easy to extend results to super polynomial time classes such
as $\NE\not\subseteq \NT(n^{O(\log n)})_{\T}(\NT(n^{O(\log n)})\cap
\sparse)$. We will present this kind of results in the extended
version of this paper.


\end{document}

\section{Open problems.}

1. Separate NEXP from $NP_{\rm T}({PI})$, where $PI$ represents the
polynomial identity problem.

2. Separate NEXP from $\P_{\rm T}({PI})$.

3. Separate NEXP from $NP_{\rm T}({GI})$, where $GI$ is the graph
isomorphism problem or group isomorphism problem.

4. Separate NEXP from $\P_{\rm T}({GI})$.

5. Change the sub-exponential function to sub-exponential in
infinite places.

\end{document}

%% file: mac90.tex
\newcounter{savenumi}

\newtheorem{theoremfoo}{Theorem}
\newenvironment{theorem}{\pagebreak[1]\begin{theoremfoo}}{\end{theoremfoo}}

\newtheorem{propositionfoo}[theoremfoo]{Proposition}

\newtheorem{lemmafoo}[theoremfoo]{Lemma}
\newenvironment{lemma}{\pagebreak[1]\begin{lemmafoo}}{\end{lemmafoo}}
\newtheorem{conjecturefoo}[theoremfoo]{Conjecture}

\newtheorem{corollaryfoo}[theoremfoo]{Corollary}
\newenvironment{corollary}{\pagebreak[1]\begin{corollaryfoo}}{\end{corollaryfoo}}
\newtheorem{exercisefoo}{Exercise}

\newtheorem{openfoo}[theoremfoo]{Question}

\newtheorem{nttn}[theoremfoo]{Notation}

\newtheorem{dfntn}[theoremfoo]{Definition}
\newenvironment{definition}{\pagebreak[1]\begin{dfntn}\rm}{\end{dfntn}}

\newenvironment{proof}
    {\pagebreak[1]{\narrower\noindent {\bf Proof:\quad\nopagebreak}}}{\QED}

\newcommand{\dash}{{\rm\mbox{-}}}
\newcommand{\floor}[1]{\left\lfloor#1\right\rfloor}
\newcommand{\ceiling}[1]{\left\lceil#1\right\rceil}

\newcommand{\DTIME}{{\rm DTIME}}

\newcommand{\NTIME}{{\rm NTIME}}

\newcommand{\NE}{{\rm NE}}

\newcommand{\E}{{\rm E}}
\newcommand{\SAT}{{\rm SAT}}
\newcommand{\NP}{{\rm NP}}

\renewcommand{\P}{{\rm P}}      
\newcommand{\coNP}{{\co\NP}}

\newcommand{\T}{{\rm T}}

\def\nre.{$n$\/-r.e.}
\newcommand{\scrod}{\quad\nopagebreak}

\newcommand{\co}{{\rm co}\dash}

\newcommand{\EXP}{{\rm EXP}}
\newcommand{\NEXP}{{\rm NEXP}}

\newtheorem{factfoo}[theoremfoo]{Fact}

\newcommand{\BPP}{{\rm BPP}}

\newtheorem{propertyfoo}[theoremfoo]{Property}

\makeatletter

\def\@makechapterhead#1{ \vspace*{50pt} { \parindent 0pt \raggedright 
 \ifnum \c@secnumdepth >\m@ne \huge\bf \@chapapp{} \thechapter. \par 
 \vskip 20pt \fi \Huge \bf #1\par 
 \nobreak \vskip 40pt } }

\def\@sect#1#2#3#4#5#6[#7]#8{\ifnum #2>\c@secnumdepth
     \def\@svsec{}\else 
     \refstepcounter{#1}\edef\@svsec{\csname the#1\endcsname.\hskip 1em }\fi
     \@tempskipa #5\relax
      \ifdim \@tempskipa>\z@ 
        \begingroup #6\relax
          \@hangfrom{\hskip #3\relax\@svsec}{\interlinepenalty \@M #8\par}
        \endgroup
       \csname #1mark\endcsname{#7}\addcontentsline
         {toc}{#1}{\ifnum #2>\c@secnumdepth \else
                      \protect\numberline{\csname the#1\endcsname}\fi
                    #7}\else
        \def\@svsechd{#6\hskip #3\@svsec #8\csname #1mark\endcsname
                      {#7}\addcontentsline
                           {toc}{#1}{\ifnum #2>\c@secnumdepth \else
                             \protect\numberline{\csname the#1\endcsname}\fi
                       #7}}\fi
     \@xsect{#5}}

\def\@begintheorem#1#2{\it \trivlist \item[\hskip \labelsep{\bf #1\ #2.}]}

\def\@opargbegintheorem#1#2#3{\it \trivlist
      \item[\hskip \labelsep{\bf #1\ #2\ (#3).}]}

\makeatother

\newif\ifshortconferences
\shortconferencesfalse
\newif\ifmediumconferences
\mediumconferencesfalse

\def\ending#1{{\count1=#1\relax
\count2=\count1
\divide\count2 by 100
\multiply\count2 by 100
\advance\count1 by -\count2
\ifnum\count1=11
th%
\else \ifnum\count1=12
th%
\else \ifnum\count1=13
th%
\else 
\count2=\count1
\divide\count1 by 10
\multiply\count1 by 10
\advance\count2 by -\count1
\ifnum\count2=1
st%
\else \ifnum\count2=2
nd%
\else \ifnum\count2=3
rd%
\else th%
\fi\fi\fi\fi\fi\fi
}}

\def\STOC{\conf{STOC}}

\def\Proceedings{\ifshortconferences Proc.\else\ifmediumconferences Proc.\else Proceedings\fi\fi}
\def\Proceedingsofthe{\ifshortconferences Proc.\else\ifmediumconferences Proc.\else Proceedings of the\fi\fi}

\newcounter{confnum}

\def\conf#1#2{%
\setcounter{confnum}{#2}%
\addtocounter{confnum}{-\csname #1zero\endcsname}%
\ifnum\value{confnum}=1%
\expandafter\ifx\csname #1One\endcsname\relax%
\Proceedingsofthe\ \arabic{confnum}\ending{\value{confnum}}\ \csname #1name\endcsname%
\else \csname #1One\endcsname\fi%
\else%
\Proceedingsofthe\
\arabic{confnum}\ending{\value{confnum}}\ \csname #1name\endcsname\fi}

\def\qsym{\vrule width0.7ex height0.9em depth0ex}

\newif\ifqed\qedtrue

\def\noqed{\global\qedfalse}

\def\qed{\ifqed{\penalty1000\unskip\nobreak\hfil\penalty50
\hskip2em\hbox{}\nobreak\hfil\qsym
\parfillskip=0pt \finalhyphendemerits=0\par\medskip}\fi\global\qedtrue}

\makeatletter
\def\eqnqed{\noqed
	\def\@tempa{equation}
	\ifx\@tempa\@currenvir\def\@eqnnum{\qsym}%
	\addtocounter{equation}{-1}\else%
    \def\@@eqncr{\let\@tempa\relax
    \ifcase\@eqcnt \def\@tempa{& & &}\or \def\@tempa{& &}%
      \else \def\@tempa{&}\fi
     \@tempa {\def\@eqnnum{{\qsym}}\@eqnnum}
     \global\@eqnswtrue\global\@eqcnt\z@\cr}\fi}

\def\eqnlabel#1#2{\if@filesw {\let\thepage\relax%
   \def\protect{\noexpand\noexpand\noexpand}%
   \edef\@tempa{\write\@auxout{\string
      \newlabel{#2}{{{#1}}{\thepage}}}}%
   \expandafter}\@tempa%
   \if@nobreak \ifvmode\nobreak\fi\fi\fi%
	\def\@tempa{equation}
	\ifx\@tempa\@currenvir\def\theequation{{#1}}%
	\addtocounter{equation}{-1}\else%
    \def\@@eqncr{\let\@tempa\relax
    \ifcase\@eqcnt \def\@tempa{& & &}\or \def\@tempa{& &}%
      \else \def\@tempa{&}\fi
     \@tempa {\def\@eqnnum{{#1}}\@eqnnum}
     \global\@eqnswtrue\global\@eqcnt\z@\cr}\fi}

\makeatother

\def\QED{\qed}

\makeatother




%% file: ne-npnp.bbl
\begin{thebibliography}{10}

\bibitem{Adleman}
L.~Adleman.
\newblock Two theorems on random polynomial time.
\newblock In {\em {\Proceedings} of the 19th Annual IEEE Symposium on
  Foundations of Computer Science}, pages 75--83, 1978.

\bibitem{AllenderBeigelHertrampfHomer}
E.~Allender, R.~Beigel, U.~Hertrampf, and S.~Homer.
\newblock Almost-everywhere complexity hierarchies for nondeterministic time.
\newblock {\em Theoretical Computer Science}, 115:225--241, Aug. 1993.

\bibitem{BermanHartmanis}
L.~Berman and J.~Hartmanis.
\newblock On isomorphism and density of {NP} and other complete sets.
\newblock {\em SICOMP}, 6:305--322, 1977.

\bibitem{BuhrmanFortnowSanthanam09}
H.~Buhrman, L.~Fortnow, and R.~Santhanam.
\newblock Unconditional lower bounds against advice.
\newblock In {\em Proceedings of the 36th International Colloquium on Automata,
  Languages and Programming, 36th International Colloquium (ICALP'09)}, pages
  195--209, 2009.

\bibitem{BuhrmanTorenvliet94}
H.~Buhrman and L.~Torenvliet.
\newblock On the cutting edge of relativization: The resource bounded injury
  method.
\newblock In {\em Proceedings of the 21st International Colloquium on Automata,
  Languages and Programming, Lecture Notes in Computer Science 820, Springer},
  pages 263--273, 1994.

\bibitem{BussHay}
S.~R. Buss and L.~E. Hay.
\newblock On truth table reducibility to {SAT}.
\newblock {\em Inf. \& Comp.}, 91(1):86--102, Mar. 1991.

\bibitem{Cook-nondet-hierarchy}
S.~Cook.
\newblock A hierarchy for nondeterministic time complexity.
\newblock {\em JCSS}, 7:343--353, 1973.

\bibitem{Fu95}
B.~Fu.
\newblock With quasi-linear queries exp is not polynomial-time turning
  reducible to sparse sets.
\newblock {\em SIAM Journal on Computing}, pages 1082--1090, 1995.

\bibitem{FuLiZhang09}
B.~Fu, A.~Li, and L.~Zhang.
\newblock Separating ne from some nonuniform nondeterministic complexity
  classes.
\newblock In {\em In Proceedings of the 15th Annual International Conference in
  Computing and Combinatorics, Lecture Notes in Computer Science 5609}, pages
  486--495, 2009.

\bibitem{FuLiZhong94}
B.~Fu, H.-Z. Li, and Y.~Zhong.
\newblock An application of the translational method.
\newblock {\em Mathematical Systems Theory}, 27:183--186, 1994.

\bibitem{HarkinsHitchcock07}
R.~Harkins and J.~Hitchcock.
\newblock Dimension, halfspaces, and the density of hard sets.
\newblock In {\em In Proceedings of the 13th Annual International Conference
  Computing and Combinatorics (COCOON 2007), Lecture Notes in Computer Science
  4598, 2007}, pages 129--139, 2007.

\bibitem{Heller}
H.~Heller.
\newblock On relativized exponential and probabilistic complexity classes.
\newblock {\em Inf. \& Comp.}, 71:231--243, 1986.

\bibitem{Hitchcock06}
J.~Hitchcock.
\newblock Online learning and resource-bounded dimension: Winnow yields new
  lower bounds for hard sets.
\newblock In {\em Proceedings of the 23rd Annual Symposium on Theoretical
  Aspects of Computer Science (STACS 2006), Lecture Notes in Computer Science
  3884}, pages 408--419, 2006.

\bibitem{ImpagliazzoPaturi99}
R.~Impagliazzo and R.~Paturi.
\newblock The complexity of k-sat.
\newblock In {\em Proceedings of the 14th IEEE Conference on Computational
  Complexity}, page 237 – 240, 1999.

\bibitem{KarpLipton}
R.~M. Karp and R.~J. Lipton.
\newblock Some connections between nonuniform and uniform complexity classes.
\newblock In {\em \STOC{80}}, pages 302--309, 1980.

\bibitem{LutzMayordomo94}
J.~Lutz and E.~Mayordomo.
\newblock Measure, stochasticity, and the density of hard languages.
\newblock {\em SIAM Journal on Computing}, 23(4):762--779, 1994.

\bibitem{Mahaney}
S.~Mahaney.
\newblock Sparse complete sets for {NP}: Solution to a conjecture of {B}erman
  and {H}artmanis.
\newblock {\em JCSS}, 25:130--143, 1982.

\bibitem{Mocas96}
S.~Mocas.
\newblock Separating classes in the exponential-time hierarchy from classes in
  ph.
\newblock {\em Theor. Comput. Sci.}, 158:221--231, 1996.

\bibitem{OgiwaraWatanabe01}
M.~Ogiwara and O.~Watanabe.
\newblock On polynomial-time bounded truth-table reducibility of np sets to
  sparse sets.
\newblock {\em SIAM J. Comput.}, 20(3):471--483, 1991.

\bibitem{SeiferasFischerMeyer78}
J.~Seiferas, M.~J. Fischer, and A.~Meyer.
\newblock Separating nondeterministic time complexity classes.
\newblock {\em Journal of ACM}, 25:146–167, 1978.

\bibitem{Watanabe87}
O.~Watanabe.
\newblock Polynomial time reducibility to a set of small density.
\newblock In {\em Proceedings of the 2nd IEEE Structure in Complexity Theory
  Conference}, pages 138--146, 1987.

\bibitem{Williams10}
R.~Williams.
\newblock Non-uniform acc circuit lower bounds.
\newblock http://www.cs.cmu.edu/~ryanw/acc-lbs.pdf, 2010.

\bibitem{Zak83}
S.~Zak.
\newblock A turing machine hierarchy.
\newblock {\em Theoretical Computer Science}, 26:327–333, 1983.

\end{thebibliography}
